\documentclass[xcolor=pdftex,dvipsnames,table]{article}
 
\usepackage{microtype}

\usepackage{amsmath,amsthm,amscd,amsfonts,amssymb,amstexalg}
\usepackage{verbatim}
\usepackage{enumerate}

\usepackage{mathrsfs} 

\DeclareMathOperator{\prob}{\mathbf{P}}

\newcommand{\ETH}{\mathbf{ETH}}

\newcommand{\SETH}{\mathbf{SETH}}

\newcommand{\ints}{\mathbb{Z}}

\theoremstyle{plain}
\newtheorem{theorem}{Theorem}[section]

\newtheorem{lemma}[theorem]{Lemma}

\newtheorem{corollary}[theorem]{Corollary}

\theoremstyle{definition}

\input local.tex

\makeatletter
\def\blfootnote{\gdef\@thefnmark{}\@footnotetext}
\makeatother

\title{On the Fine-grained Complexity of One-Dimensional Dynamic Programming}
\author{Marvin K\"unnemann \qquad Ramamohan Paturi \qquad Stefan Schneider
\\University of California, San Diego}

\newcommand{\myparagraph}[1]{\vspace{0.4em}{\noindent \bfseries #1}}

\begin{document}
\blfootnote{This research is supported by the Simons Foundation. This research is support
ed by NSF
grant CCF-1213151 from the
Division  of Computing and Communication Foundations.
Any opinions,
findings and conclusions or
recommendations expressed in this material are those
of the authors and do
not necessarily reflect the
views of the National Science Foundation.}

\maketitle

\begin{abstract}
  In this paper, we investigate the complexity of
  \emph{one-dimensional dynamic programming}, or more specifically,
  of the Least-Weight Subsequence (\LWS) problem: Given a sequence of
  $n$ data items together with weights for every pair of the items,
  the task is to determine a subsequence~$S$ minimizing the total
  weight of the pairs adjacent in $S$. A large number of natural
  problems can be formulated as \LWS problems, yielding obvious
  $\Oh(n^2)$-time solutions.

  In many interesting instances, the $\Oh(n^2)$-many weights can be
  succinctly represented. Yet except for near-linear time algorithms
  for some specific special cases, little is known about when an \LWS
  instantiation admits a subquadratic-time algorithm and when it does
  not. In particular, no lower bounds for \LWS instantiations have
  been known before. In an attempt to remedy this situation, we
  provide a general approach to study the fine-grained complexity of
  succinct instantiations of the LWS problem. In particular, given an
  \LWS instantiation we identify a highly parallel \emph{core} problem
  that is subquadratically \emph{equivalent}.  This provides either an
  explanation for the apparent hardness of the problem or an avenue to
  find improved algorithms as the case may be.

  More specifically, we prove subquadratic equivalences between the
  following pairs (an \LWS instantiation and the corresponding core
  problem) of problems: a low-rank version of \LWS and minimum inner
  product, finding the longest chain of nested boxes and vector
  domination, and a coin change problem which is closely related to
  the knapsack problem and \mpconv.  Using these equivalences and
  known $\SETH$-hardness results for some of the core problems, we
  deduce tight conditional lower bounds for the corresponding \LWS
  instantiations. We also establish the \mpconv-hardness of the
  knapsack problem. Furthermore, we revisit some of the \LWS
  instantiations which are known to be solvable in near-linear time
  and explain their easiness in terms of the easiness of the
  corresponding core problems.
\end{abstract}

\section{Introduction}
Dynamic programming (DP) is one of the most fundamental paradigms for
designing algorithms and a standard topic in textbooks on algorithms.
Scientists from various disciplines have developed DP formulations for
basic problems encountered in their applications.  However, it is not
clear whether the existing (often simple and straightforward) DP
formulations are in fact optimal or nearly optimal.  Our lack of
understanding of the optimality of the DP formulations is particularly
unsatisfactory since many of these problems are computational
primitives.

Interestingly, there have been recent developments regarding the
optimality of standard DP formulations for some specific problems,
namely, conditional lower bounds assuming the Strong Exponential Time
Hypothesis ($\SETH$)~\cite{ImpagliazzoPaturi01}. The longest common
subsequence (LCS) problem is one such problem for which almost tight
conditional lower bounds have been obtained recently. The LCS problem
is defined as follows: Given two strings $x$ and $y$ of length at most
$n$, compute the length of the longest string $z$ that is a
subsequence of both $x$ and $y$. The standard DP formulation for the
LCS problem involves computing a two-dimensional table requiring
$\Oh(n^2)$ steps. This algorithm is only slower than the fastest known
algorithm due to Masek and Paterson~\cite{MasekP80} by a
polylogarithmic factor. However, there has been no progress in finding
more efficient algorithms for this problem since the 1980s, which
prompted attempts as early as in 1976~\cite{AhoHu76} to understand the
barriers for efficient algorithms and to prove lower
bounds. Unfortunately, there have not been any nontrivial
unconditional lower bounds for this or any other problem in general
models of computation. This state of affairs prompted researchers to
consider \emph{conditional} lower bounds based on conjectures such as
3-Sum conjecture~\cite{GajentaanO95} and more recently based on
$\ETH$~\cite{ImpagliazzoPZ01} and $\SETH$~\cite{ImpagliazzoPaturi01}.
Researchers have found $\ETH$ and $\SETH$ to be useful to explain the
exact complexity of several $\NP$-complete problems (see the survey
paper~\cite{LokshtanovMS11}). Surprisingly, Ryan Williams~\cite{Williams05} has found a simple reduction from the $\cnfsat$
problem to the orthogonal vectors problem which under $\SETH$ leads to
a matching quadratic lower bound for the orthogonal vectors
problem. This in turn led to a number of conditional lower bound
results for problems in $\P$ (including LCS and related problems)
under $\SETH$~\cite{BackursI15,
  AbboudBVW15,BringmannK15,AbboudHVWW16,GaoIKW17}.  Also see
\cite{VassilevskaWilliams15} for a recent survey.

The DP formulation of the LCS problem is perhaps the conceptually
simplest example of a \emph{two-dimensional} DP formulation.  In the
standard formulation, each entry of an $n \times n$ table is computed
in constant time. The LCS problem belongs to the class of alignment
problems which, for example, are used to model similarity between gene
or protein sequences.  Conditional lower bounds have recently been
extended to a number of alignment
problems~\cite{Bringmann14,BackursI15,AbboudBVW15,BringmannK15, AbboudVWW14}.

In contrast, there are many problems for which natural 
quadratic-time DP formulations compute a \emph{one-dimensional} table  
of length $n$ by
spending $\Oh(n)$-time per entry. In this work, we investigate the
optimality of such DP formulations and obtain 
new (conditional) lower bounds which match the complexity of the standard
DP formulations.

\myparagraph{1-dimensional DP: The 
Least-Weight Subsequence (\LWS) Problem.} 
In this paper, we investigate the optimality of the standard DP formulation
of the \LWS problem.
A classic example of an \LWS problem
is airplane refueling~\cite{HirschbergL87}: Given airport
locations on a line, and a preferred distance per hop $k$ (in miles),
we define the penalty for flying $k'$ miles as $(k - k')^2$. The goal
is then to find a sequence of airports terminating at the last airport
that minimizes the sum of the penalties.
We now define the \LWS problem formally.

\begin{problem}[\LWS]
\label{prob:lws}
We are given a sequence of $n+1$ data items $x_0,\dots,x_n$,  
weights $w_{i,j} \in \{-W,\dots,W\} \cup \{\infty\}$ for every 
pair $i<j$ of indices where the weights may also be functions of 
the values  of  data items $x_i$, 
and an arbitrary function $g: \ints \to \ints$.  The \LWS problem is 
to determine $T[n]$ which is defined by the following DP formulation.
\begin{align}
T[0] & = 0, \nonumber\\
T[j] & = \min_{0\le i < j} g(T[i]) + w_{i,j} \qquad \text{for } j = 1,\dots,n. \label{eq:rec}
\end{align}
\end{problem}

To formulate  airplane refueling as an \LWS problem, 
we let $x_i$ be the location of the $i$'th airport, 
$g$ be the identity function, and
$w_{i,j} = (x_j - x_i - k)^2$.

In the definition of the \LWS problem, we  did not
specify the encoding of the problem (in
particular, the type of data items and the representation of
the weights $w_{i,j}$) so we can capture a larger 
variety of problems: it not only encompasses
classical problems such as the pretty printing problem due to Knuth and
Plass~\cite{KnuthP81}, the airplane refueling
problem~\cite{HirschbergL87} and the longest increasing subsequence
(LIS)~\cite{Fredman75}, 
but also the unbounded 
subset sum problem~\cite{Pisinger03, Bringmann17},
a more general coin change problem that is
effectively equivalent to the unbounded knapsack problem,
1-dimensional $k$-means clustering problem~\cite{GronlundLMN17}, 
finding longest $\rel$-chains (for an arbitrary binary relation $\rel$),
and many others 
(for a more complete list of problems definitions, see
Section~\ref{sec:prelim}). 

Under mild assumptions on the encoding of the data items and weights,
any instantiation of the \LWS problems
can be solved in time $\Oh(n^2)$ using~\eqref{eq:rec} 
for determining the values $T[j], j=1,\dots,n$ 
in
time $\Oh(n)$ each. 
However, the best known algorithms for the \LWS  problems differ
quite significantly in their time complexity.
Some problems including the pretty printing, airline
refueling and LIS turn out to be solvable in near-linear time, while
no subquadratic algorithms are known for  the
unbounded knapsack problem or for finding
the longest $\rel$-chain.

The main goal of the paper is to investigate the optimality of
the \LWS DP formulation for various problems by  proving
conditional lower bounds.

\myparagraph{Succinct \LWS instantiations.} 
In the extremely long presentation of an \LWS problem,
the weights $w_{i,j}$ are given
explicitly. This is however not a very interesting case 
from a computational point of view, 
as the standard DP formulation takes linear time (in the size of the
input) to compute $T[n]$. 
In the example of the airplane refueling problem the size of the input
is only $\Oh(n)$ assuming that the values of the data items are
bounded by some polynomial in $n$.  For such succinct representations,
we ask if the quadratic-time algorithm based on the standard \LWS DP
formulation is optimal.  Our approach is to study several natural
succinct versions of the \LWS problem (by specifying the type of data
items and the weight function\footnote{In all our applications, the
  function $g$ is the trivial identity function.}) and determine their
complexity.  We refer to Section \ref{sec:prelim} for 
examples of succinct instantiations of the \LWS problem.

\myparagraph{Our Contributions and Results.} 
The main contributions of our paper include a general
framework for reducing succinct \LWS instantiations
to what we call the \emph{core} problems
and proving subquadratic equivalences between them.
The subquadratic equivalences are 
interesting for two reasons. First, they allows us
to conclude conditional lower bounds for certain 
\LWS instantiations, where
previously no lower bounds are known. Second, subquadratic (or
more general fine-grained) equivalences are more useful since they let
us translate hardness as well as easiness results.

Our results include tight (up to subpolynomial factors) conditional lower
bounds for several \LWS instantiations with succinct
representations. These instantiations include the coin change problem,
low rank versions of the \LWS problem, and the longest subchain
problems. Our results are somewhat more general. We propose a
\emph{factorization} of the \LWS problem into a \emph{core} problem
and a fine-grained reduction from the \LWS problem to the core
problem. The idea is that core problems (which are often well-know
problems) capture the hardness of the \LWS problem and act as a
potential barrier for more efficient algorithms. While we do not
formally define the notion of a core problem, we identify several core
problems which share several interesting properties. For example,
they do not admit natural DP formulations and are easy to parallelize.
In contrast, the quadratic-time DP formulation of \LWS problems
requires the entries $T[i]$ to be computed in order, suggesting that
the general problem might be inherently sequential.

The reductions between LWS problems and core problems involve a
natural intermediate problem, which we call the \StaticLWS problem.
We first reduce the \LWS problem to the \StaticLWS problem in a
general way and then reduce the \StaticLWS problem to a core problem.
The first reduction is divide-and-conquer in nature and is inherently
sequential. The latter reduction is specific to the instantiation of
the LWS problem. The \StaticLWS problem is easy to parallelize and
does not have a natural DP formulation. However, the problem is not
necessarily a natural problem. The \StaticLWS problem can be thought
of as a generic core problem, but it is output-intensive.

In the other direction, we show that many of the core problems can be
reduced to the corresponding \LWS instantiations thus establishing
an equivalency between LWS instantiations and their core problems. 
This equivalence enables us to translate both the hardness and 
easiness results (i.e., the subquadratic-time algorithms) for the 
core problems to the corresponding LWS instantiations.

The first natural succinct representation of the \LWS problem we consider is the
low rank \LWS problem, where the weight matrix
$\weightmat = (w_{i,j})$ is of low rank and thus representable
as $\weightmat = L\cdot R$ where $L$ and $ R^\mathrm{T}$ are
$(n\times n^{o(1)})$-matrices. For this low rank \LWS problem,
we identify
the minimum inner product problem (\MinInnProd) as a suitable core
problem. 
It is only natural and not particularly surprising that
\MinInnProd can be reduced to the low-rank LWS problem which
shows the $\SETH$-hardness of 
the low-rank \LWS problem.  The other direction is
more surprising: Inspired by an elegant trick of Vassilevska Williams
and Williams~\cite{VassilevskaWW10}, we are able to show a
subquadratic-time reduction from the (highly sequential) low-rank
\LWS problem to the (highly parallel) \MinInnProd problem. Thus, the very
compact problem $\MinInnProd$ problem captures exactly the complexity of
the low-rank LWS problem (under subquadratic reductions).

We also show that the coin change problem is
subquadratically equivalent to the \mpconv problem. In the coin change  problem,
the weight matrix $\weightmat$ is succinctly given as a Toeplitz
matrix.
At this point, the conditional hardness of the \mpconv problem is unknown.
The quadratic-time hardness of the \mpconv problem would be very interesting,
since it is known that the \mpconv problem is reducible to the 3-sum
problem and the APSP problem, 
However, recent results give surprising subquadratic-time
algorithms for special cases of \mpconv~\cite{ChanL15}. If these subquadratic-time
algorithms extend to the general \mpconv problem, 
our equivalence
result also provides a subquadratic-time algorithm for the coin change problem
and the closely related 
unbounded knapsack problem. As a corollary, our reductions also give a
quadratic-time \mpconv-based lower bound for the bounded case of
knapsack.

We next consider
the problem of finding longest chains: here, we search for the longest subsequence
(chain) in the input sequence such that all adjacent pairs in the
subsequence are contained in some binary relation $\rel$.  
We show that for any binary
relation $\rel$ satisfying certain conditions the chaining problem is
subquadratically equivalent to a corresponding (highly parallel)
selection problem. As corollaries, we get equivalences between finding
the longest chain of nested boxes (\NestedBoxes) and \VectorDomination
as well as between
finding the longest subset chain (\SubsetChain) and the orthogonal
vectors (\OV) problem. Interestingly, these results have algorithmic
implications: known algorithms for low-dimensional vector domination
and low-dimensional orthogonal vectors translate to faster algorithms
for low-dimensional \NestedBoxes and \SubsetChain for small universe
size.

Table~\ref{tab:results} lists the \LWS succinct instantiations (as
discussed above) and their corresponding core problems.  All LWS
instantiations and core problems considered in this paper are formally
defined in Section~\ref{sec:prelim}.

Finally, we revisit classic problems including the longest increasing
subsequence problem, the unbounded subset sum problem and the
concave \LWS problem and analyze 
the \StaticLWS instantiations to immediately infer that the corresponding
core problem can be solved in
near-linear time. Table~\ref{tab:positiveResults} gives an overview of
some of the problems we look at in this context.

\begin{table}
\begin{tabular}{l|lll}
\textbf{Name} & \textbf{Weights} & \textbf{Equivalent Core} & \textbf{Reference} \\
\hline
\hline
Coin Change & Toeplitz matrix: & \mpconv & Theorem~\ref{thm:coinchange} \\
& $w_{i,j} = w_{j-i}$ & & \\
& \multicolumn{3}{l}{\textbf{Remark:} Subquadratically equivalent to \UnboundedKnapsack} \\
\hline
\hline
LowRankLWS & Low rank representation:& \MinInnProd & Theorem~\ref{thm:lrlws} \\
& $w_{i,j} = \langle \sigma_i, \mu_j \rangle$ & & \\
\hline
\hline
$\rel$-chains & matrix induced by $\rel$: & $\Selection(\rel)$ & Theorem~\ref{thm:LWStoSel} \\
& $w_{i,j} = w_j$ if $\rel(x_i,x_j)$ and $\infty$ o/w & & Theorem~\ref{thm:SelToLWS}\\
& \multicolumn{3}{l}{\textbf{Remark:} Result below are corollaries.} \\
\hline
NestedBoxes & $w_{i,j} = -1$ if $B_j$ contains $B_i$ & \textsc{VectorDomination} & \\
SubsetChain & $w_{i,j} = -1$ if $S_i \subseteq S_j$ & \textsc{OrthogonalVectors} & \\
\end{tabular}
\caption{Summary of our results}
\label{tab:results}
\end{table}

\begin{table}
\begin{tabular}{l|llll}
\textbf{Name} & \textbf{Weights} & \textbf{$\tOh(n)$-time reducible to} & \textbf{Reference}\\
\hline
\hline
Longest Increasing & matrix induced by $\rel_<$: & \Sorting & \cite{Fredman75},\\
 Subsequence & $w_{i,j} = -1$ if $x_i < x_j$ & &   Observation~\ref{obs:lis} \\
\hline
Unbounded Subset & Toeplitz $\{0,\infty\}$ matrix: & \convtext & \cite{Bringmann17}, \\
 Sum & $w_{i,j} = w_{j-i} \in \{0, \infty\}$ & & Observation~\ref{obs:subsetsum} \\
\hline
Concave 1-dim. DP & concave matrix: & SMAWK problem & \cite{HirschbergL87, GalilP90,Wilber88}, \\
& $w_{i,j} + w_{i',j'} \le w_{i',j} + w_{i, j'}$ & & Observation~\ref{obs:concavelws}  \\
& for $i\le i'\le j \le j'$ & & \\
\end{tabular}
\caption{Near-linear time algorithms following from the proposed framework.}
\label{tab:positiveResults}
\end{table}

\myparagraph{Related Work.}  LWS has been introduced by Hirschberg and
Lamore~\cite{HirschbergL87}. If the weight function satisfies the
\emph{quadrangle inequality}\footnote{See Section~\ref{sec:prelim} for
  definitions.} formalized by Yao~\cite{Yao80}, one obtains the
\emph{concave LWS} problem, for which they give an
$\Oh(n \log n)$-time algorithm. Subsequently, improved algorithms
solving concave LWS in time $\Oh(n)$ were given~\cite{Wilber88,
  GalilP90}. This yields a fairly large class of weight functions
(including, e.g., the pretty printing and airplane refueling problems)
for which linear-time solutions exist. To generalize this class of
problems, further works address convex weight functions\footnote{A
  weight function is convex if it satisfies the inverse of the
  quadrangle inequality.}~\cite{GalilG89, MillerM88,KlaweK90} as well
as certain combinations of convex and concave weight
functions~\cite{Eppstein90} and provide near-linear time
algorithms. For a more comprehensive overview over these algorithms
and further applications of the LWS problem, we refer the reader to
Eppstein's PhD thesis~\cite{EppsteinPhD}.

Apart from these notions of concavity and convexity, results on the
succinct LWS problems are typically more scattered and
problem-specific (see, e.g.,
\cite{Fredman75,KnuthP81,Bringmann17,GronlundLMN17}; furthermore, a
closely related recurrence to \eqref{eq:rec} pops up when solving
bitonic TSP~\cite{deBergBJW16}). An exception to this rule is a study
of the parallel complexity of LWS~\cite{GalilP94}.

{\myparagraph{Organization.}}  Section \ref{sec:lrlws} contains the
result on low-rank \LWS. This is also where we formally introduce
\StaticLWS. Section \ref{sec:coinchange} proves the subquadratic
equivalence of the coin change problem and \mpconv, while Section
\ref{sec:chains} discusses chaining problems and their corresponding
selection (core) problem. Our results
on near-linear time algorithms are given in Section \ref{sec:nearlinear}.

\section{Preliminaries}
\label{sec:prelim}

In this section, we state our notational conventions and list the main problems considered in this work.

Problem $A$ \emph{subquadratically reduces} to problem $B$, denoted
$A \lequad B$, if for any $\varepsilon >0$ there is a $\delta > 0$
such that an algorithm for $B$ with time $\Oh(n^{2-\varepsilon})$
implies an algorithm for $A$ with time $\Oh(n^{2-\delta})$. We call
the two problems subquadratically equivalent, denoted $A \eqquad B$, if
there are subquadratic reductions both ways.

We let $[n]:=\{1,\dots,n\}$.
When stating running time, we use the notation $\tOh(\cdot)$ to hide polylogarithmic factors. For a problem $P$, we write $T^P$ for its time complexity.
We generally assume the word-RAM model of computation with word size $w = \Theta(\log n)$. For most problems defined in this paper, we consider inputs to be integers in the range
$\{-W, \ldots, W\}$ where $W$ fits in a constant number of
words\footnote{For the purposes of our reductions, even values up to
  $W= 2^{n^{o(1)}}$ would be fine.}. For vectors, we use $d$ for the
dimension and generally assume $d=n^{o(1)}$.

\myparagraph{Core Problems and Hypotheses.}
One of the most popular problems in the field of quadratic-time conditional hardness is the following problem.

\begin{problem}[Orthogonal Vectors (\OV)]
  Given $a_1, \dots, a_n, b_1,\dots,b_n \in \{0,1\}^d$, determine
  if there is a pair $i,j$ satisfying $\langle a_i,b_j \rangle =
  0$.
\end{problem}

Recall that for \OV (and the related problems below) we assume $d=n^{o(1)}$. Thus the naive algorithm solves \OV in time $\Oh(n^2 \cdot d) = \Oh(n^{2+o(1)})$.

One of the reasons for the popularity of \OV is its surprising connection to the Strong Exponential Time Hypothesis ($\SETH$)~\cite{ImpagliazzoPaturi01}: 
It states that for every $\varepsilon > 0$ there
is a $k$, such that the $k$-SAT problem requires time
$\Omega(2^{(1-\varepsilon) n})$.  
By an elegant reduction due to Williams~\cite{Williams05}, \OV is quadratic-time $\SETH$-hard, i.e., there is no algorithm with running time
time $\Oh(n^{2-\varepsilon})$ for any $\varepsilon >0$ unless $\SETH$ is false.

We consider the following generalizations of \OV.

\begin{restatable}[\MinInnProd]{problem}{MinInnProdProb}
  Given $a_1, \dots, a_n, b_1,\dots,b_n \in \{-W,\dots,W\}^d$ and a
  natural number $r\in \ints$, determine if there is a pair $i,j$
  satisfying $\langle a_i,b_j \rangle \le r$.
\end{restatable}

\begin{problem}[\AllInnProd]
  Given $a_1,\dots, a_n \in \{-W,\dots,W\}^d$ and
  $b_1, \dots, b_n \in \{-W,\dots,W\}^d$, determine \emph{for all
    $j\in [n]$}, the value $\min_{i\in[n]} \langle a_i, b_j
  \rangle$.
\end{problem}

\begin{problem}[\VectorDomination]
  Given $a_1, \dots, a_n, b_1,\dots,b_n \in \{-W,\dots,W\}^d$
  determine if there is a pair $i,j$ such that $a_i \leq b_j$
  component-wise.
\end{problem}

\begin{problem}[\SetContainment]
  Given sets $a_1, \dots, a_n, b_1,\dots,b_n \subseteq [d]$ given as
  vectors in $\{0,1\}^d$
  determine if there is a pair $i,j$ such that $a_i \subseteq b_j$.
\end{problem}

Note that \SetContainment is a special case of \VectorDomination and
computationally equivalent to \OV, as $\langle a, b \rangle = 0$ if
and only if $a \subseteq \overline{b}$ (in this slight misuse of
notation we think of the Boolean vectors $a,b$ as sets and let $\bar{b}$ denote the complement of $b$).

Since subquadratic solutions to any of these problems trivially give a
subquadratic solution to $\OV$, these problems are also quadratic-time
$\SETH$-hard. However, the converse does not necessarily hold. In
particular, the strongest currently known upper bounds differ: while
for $\OV$ and \SetContainment for small dimension $d=c \cdot \log(n)$, an $n^{2-1/\Oh(\log c)}$-time
algorithm is known~\cite{AbboudWY15}, for \VectorDomination the best
known algorithm runs only in time
$n^{2-1/\Oh( c \log^2 c)}$~\cite{ImpagliazzoLPS14,Chan15}.

Another fundamental quadratic-time problem is \mpconv, defined below.

\begin{restatable}[\mpconv]{problem}{mpconvproblem}
  Given vectors $a = (a_0,\dots, a_{n-1})$, $b=(b_0, \dots, b_{n-1}) \in \{-W, \dots, W\}^n$, determine its \mpconv $a\conv b$ defined by
  \[ (a\conv b)_k = \min_{0 \le i, j < n: i+j = k} a_i + b_j \qquad \text{for all } 0\le k \le 2n-2. \]
\end{restatable}

As opposed to the classical convolution, which we denote as $a\classicconv b$, solvable in time $\Oh(n\log n)$ using FFT, no strongly subquadratic algorithm for \mpconv is known. Compared to \OV, we have less support for believing that no $\Oh(n^{2-\varepsilon})$-time algorithm for \mpconv exists. In particular, interesting special cases can be solved in subquadratic-time~\cite{ChanL15} and there are subquadratic-time co-nondeterministic and nondeterministic algorithms~\cite{BremnerCDEHILPT14, CarmosinoGIMPS16}. At the same time, breaking this long-standing quadratic-time barrier is a prerequisite for progress on refuting the 3SUM and APSP conjectures. This makes it an interesting target particularly for proving subquadratic \emph{equivalences}, since both positive and negative resolutions of this open question appear to be reasonable possibilities.

\myparagraph{Succinct LWS Versions and Applications.}
In the definition of LWS (Problem~\ref{prob:lws}) we did not fix the encoding of the problem (in particular, the choice of data items, as well the representation of the weights $w_{i,j}$ and the function $g$).
Assuming that $g$ can be determined in $\tOh(1)$ and that $W = \poly(n)$, this problem can naturally be solved in time $\tOh(n^2)$, by evaluating the central recurrence \eqref{eq:rec} for each $j=1,\dots,n$ -- this takes $\tOh(n)$ time for each $j$, since we take the minimum over at most $n$ expressions that can be evaluated in time $\tOh(1)$ by accessing the previously computed entries $T[0], \dots, T[j-1]$ as well as computing $g$.
In all our applications, $g$ will be the identity function, hence it will suffice to define the type of data items and the corresponding weight matrix. Throughout this paper,  whenever we fix a representation of the weight matrix $\weightmat = (w_{i,j})_{i,j}$, we denote the corresponding problem $\LWS(\weightmat)$. 

In the remainder of this section, we list problems considered in this paper that can be expressed as an LWS instantiations. At this point, we typically give the most natural formulations of these problems -- the corresponding definitions as LWS instantiations are given in the corresponding sections.

We start off with a natural succinct ``low-rank'' version of \LWS.

\begin{problem}[\LRLWS]
  \LRLWS is the \LWS problem where the weight matrix $\weightmat$ is
  of rank $d \ll n$. The input is given succinctly as two matrices $A$
  and $B$, which are $(n\times d)$- and $(d\times n)$-matrices
  respectively, and $\weightmat = A\cdot B$. 
\end{problem}

Alternatively, \LRLWS may be interpreted in the following way: There are places $0, 1, \dots, n$, each of which is equipped with an in- and an out-vector. The cost of going from place $i$ to $j$ is then defined as the inner product of the out-vector of $i$ with the in-vector of $j$, and the task is to compute the minimum-cost monotonically increasing path to reach place $n$ starting from 0. In Section~\ref{sec:lrlws}, we prove subquadratic equivalence to \MinInnProd.

We consider the following coin change problem and variations of \Knapsack.

\begin{problem}[\CoinC]
  We are given a weight sequence $w = (w_1, \dots, w_n)$ with
  $w_i \in \{-W, \dots, W\} \cup \{\infty\}$, i.e., the coin with value
  $i$ has weight $w_i$. Find the weight of the multiset of
  denominations $I$ such that $\sum_{i\in I} i= n$ and the sum of the
  weights $\sum_{i \in I} w_i$ is minimized.
\end{problem}

\begin{restatable}[\UKnapsack]{problem}{UKnapsackproblem}
  We are given a sequence of profits $p = (p_1, \dots, p_n)$ with
  $p_i \in \{0, 1, \dots, W\}$, i.e., the item of size~$i$
  has profit $p_i$. Find the total profit of the multiset of indices~$I$
  such that $\sum_{i\in I} i \leq n$ and the total profit
  $\sum_{i \in I} p_i$ is maximized.
\end{restatable}
Note that if we replace multiset by set in the above definition, we obtain the bounded version of the problem, which we denote by $\Knapsack$.

We remark that our perspective on \CoinC and \UKnapsack (as well as \USubsetSum below) using LWS is slightly different than many classical accounts of Knapsack: We define the problem size as the budget size instead of the number of items, thus our focus is on pseudo-polynomial time algorithms for the typical formulations of these problems.

Note that we state the coin change problem as allowing positive or
negative weights, but \UKnapsack only allows for positive
profits. Furthermore, \CoinC is a minimization problem, while
\UKnapsack is a maximization problem. For \CoinC, the maximization
problem is trivially equivalent as we can negate all weights. Furthermore, we can
freely translate the range of the weights in the coin change problem by
defining $w'_i = i\cdot M + w_i$ for all $i$ and sufficiently large or small $M$. The
most significant difference between \CoinC and \UKnapsack is that for
\CoinC the indices have to sum to exactly $n$, while for \UKnapsack
$n$ is only an upper bound.

We will encounter an important generalization of the two problems above, defined as follows.

\begin{problem}[\oiCoinC] 
  The output-intensive version of \CoinC is to determine, given an
  input to $\CoinC$, the weight of the optimal multiset such that the
  denominations sum up to $j$ \emph{for all $1 \leq j \leq n$}.
\end{problem}

It is easy to see that \oiCoinC is at least as hard as both \CoinC and
\UKnapsack. We will relate the above \Knapsack variants to \mpconv in Section~\ref{sec:coinchange}.

In Section~\ref{sec:nearlinear}, we will revisit near-linear time algorithms for the following special case of the coin change problem.

\begin{problem}[\USubsetSum]
  Given a subset $S\subseteq[n]$, determine whether there is a multiset of elements of $S$ that sums up to exactly $n$.
\end{problem}

We also discuss problems where the goal is to find the
longest chain among data items, where the notion of a chain is defined by some binary relation $\rel$.
We first give the definition of the general problem which is parameterized by $\rel$.

\begin{problem}[\CLWS]
  Fix a set $X$ of objects and a relation $\rel \subseteq X\times X$. The Weighted Chain Least-Weight
  Subsequence Problem for $\rel$, denoted $\CLWS(\rel)$, is the following problem: Given data items $x_0,\ldots,x_n\in X$, weights
  $w_1, \ldots, w_{n-1} \in \{-W,\ldots,W\}$, find
  the weight of the increasing sequence
  $i_0 = 0 < i_1 < i_2 < \ldots < i_k = n$ such that for all $j$ with
  $1 \leq j \leq k$ the pair $(x_{i_{j-1}},x_{i_j})$ is in the
  relation $\rel$ and the weight $\sum_{j=1}^{k-1} w_{i_j}$ is
  minimized.
\end{problem}

The following problems are specializations of this problem for different relations. 

\begin{problem}[\NestedBoxes]
  Given $n$ boxes in $d$ dimensions, given as non-negative,
  $d$-dimensional vectors $(b_1, \ldots, b_n)$, find the longest chain
  such that each box fits into the next (without rotation). We say box
  that box $a$ fits into box $b$ if for all dimensions
  $1 \leq i \leq d$, $a_i \leq b_i$.
\end{problem}

\begin{problem}[\SubsetChain]
  Given $n$ sets from a universe $U$ of size $d$, given as Boolean,
  $d$-dimensional vectors $(b_1, \ldots, b_n)$, find the longest chain
  such that each set is a subset of the next.
\end{problem}

Note that \SubsetChain is a special case of \NestedBoxes.

\begin{problem}[\LIS]
  Given a sequence of $n$ integers $x_1,\dots,x_n$, compute the length of the longest subsequence that is strictly increasing.
\end{problem}

Finally, we will briefly discuss the following class of LWS problems that turn out to be solvable in near-linear time.

\begin{problem}[\ConcaveLWS]
  Given an LWS instance in which the weights satisfy the quadrangle inequality
\[ w_{i,j} + w_{i',j'} \le w_{i',j} + w_{i, j'}  \qquad \text{for } i\le i'\le j \le j', \]
  solve it. The weights are not explicitly given, but each $w_{i,j}$ can be queried in constant time.
\end{problem}

\section{\LRLWS}
\label{sec:lrlws}

Let us first analyze the following canonical succinct representation of a low-rank weight matrix $\weightmat = (w_{i,j})_{i,j}$: If $\weightmat$ is of rank $d \ll n$, we can write it more succinctly as $\weightmat = A\cdot B$, where $A$ and $B$ are $(n\times d)$-  and $(d\times n)$ matrices, respectively. We can express the resulting natural LWS problem equivalently as follows.

\begin{problem}[\LRLWS]  We define the following \LWS instantiation $\LRLWS = \LWS(\weightmat_\lowrank)$.\\
\Data{out-vectors $\mu_0, \dots, \mu_{n-1}\in \{-W,\dots,W\}^d$, in-vectors $\sigma_1, \dots, \sigma_n \in \{-W,\dots,W\}^d$} 
\Weights{$w(i,j) = \langle \mu_i, \sigma_j \rangle$ for $0\le i < j \le n$}
\end{problem}

In this section, we show that this problem is equivalent, under subquadratic reductions, to the following \emph{non-sequential} problem.

\begin{problem}[\MinInnProd]
Given $a_1, \dots, a_n, b_1,\dots,b_n \in \{-W,\dots,W\}^d$ and a natural number $r\in \ints$, determine if there is a pair $i,j$ satisfying $\langle a_i,b_j \rangle \le r$.
\end{problem}

We first give a simple reduction from \MinInnProd that along the way
proves quadratic-time SETH-hardness of \LRLWS.

\begin{lemma}
  \label{lem:LRLWS-lb}
  It holds that $T^\MinInnProd(n,d,W) \le T^\LRLWS(2n+1, d+2,dW) + \Oh(nd)$.
\end{lemma}
\begin{proof}
Given $a_1,\dots,a_n,b_1,\dots,b_n \in \{-W,\dots,W\}^d$, let $\allzero = (0,\dots,0)\in \ints^d$ be the all-zeroes vector and define the following in- and out-vectors
\begin{align*}
\mu_0 & = (dW, 0, \allzero),  & \sigma_{2n+1} &= (dW, dW, \allzero), & \\
\mu_i & = (0, dW,  a_i),  & \sigma_i &= (0, 0, \allzero), & \text{for } i=1,\dots,n, \\
\mu_{n+j} & = (0, 0, \allzero),  & \sigma_{n+j} &= (dW, 0,  b_j), & \text{for } j=1,\dots,n.
\end{align*}
To prove correctness, we show that in the constructed \LRLWS instance, 
we have $T[2n+1] = \min_{i,j} \langle a_i, b_j \rangle$, 
from which the results follows immediately. 
Inductively, we have $T[i] = 0$ 
for $i=1,\dots,n$, since $\langle \mu_{i'},\sigma_i\rangle = 0$ 
for all $0 \le i' < i \le n$. Similarly, for $j=1,\dots,n$ 
one can  inductively show  that 
$T[n+j] = \min_{1\le i \le n, j'\le j} \langle a_i,b_{j'} \rangle$, 
using that $\langle \mu_0, \sigma_{n+j}\rangle = (dW)^2 \ge \max_{i,j} \langle a_i,b_j \rangle$, 
$\langle \mu_i, \sigma_{n+j} \rangle = \langle a_i, b_j \rangle$ and 
$\langle \mu_{n+j'}, \sigma_{n+j}\rangle = 0$  
for all $1\le i,j \le n$ and $j'\le j$. 
Finally, using (1) 
$\langle \mu_0, \sigma_{2n+1} \rangle = (dW)^2 \ge \max_{i,j} \langle a_i,b_j \rangle$ 
and $T[0]=0$, 
(2) $\langle \mu_i, \sigma_{2n+1} \rangle = (dW)^2 \ge \max_{i,j} \langle a_i, b_j \rangle$ 
and $T[i] = 0$ for $i=1,\dots,n$ 
and (3) $\langle \mu_{n+j}, \sigma_{2n+1}\rangle = 0$ 
and $T[n+j] = \min_{1\le i \le n, 1\le j' \le j} \langle a_i, b_{j'} \rangle$ 
for all $j=1,\dots,n$, 
we can finally determine $T[2n+1] = \min_{i,j} \langle a_i, b_j \rangle$.
\end{proof}

To prove the other direction, we will give a quite general approach to
compute the sequential \LWS problem by reducing to a natural static
subproblem of \LWS:

\begin{problem}[$\StaticLWS(\weightmat)$]
Fix an instance of $\LWS(\weightmat)$. Given intervals $I:=\{a+1, \dots, a+N\}$ and $J:=\{a+N+1,\dots, a+2N\}$, together with the correctly computed values $T[a+1],\dots,T[a+N+1]$, the Static Least-Weight Subsequence Problem (\StaticLWS) asks to determine
\begin{align*}
T'[j] & := \min_{i \in I} T[i] + w_{i,j}  & \text{for all } j\in J. 
\end{align*}
\end{problem}

\begin{lemma}[$\LWS(\weightmat) \lequad \StaticLWS(\weightmat)$]\label{lem:dynToStatic}
For any choice of $\weightmat$, if $\StaticLWS(\weightmat)$ can be solved in time $\Oh(N^{2-\varepsilon})$ for some $\varepsilon>0$, then $\LWS(\weightmat)$ can be solved in time $\tOh(n^{2-\varepsilon})$.
\end{lemma}
\begin{proof}
In what follows, we fix $\LWS$ as $\LWS(\weightmat)$ and $\StaticLWS$ as $\StaticLWS(\weightmat)$.

We define the subproblem $S(\{i,\dots,j\}, (t_i,\dots,t_j))$ that given an interval spanned by $1\le i \le j \le n$ and values $t_k = \min_{0\le k' < i} T[k'] + w_{k',k}$ for each point $k \in \{i,\dots, j\}$, computes all values $T[k]$ for $k \in \{i,\dots,j\}$. Note that a call to $S([n], (w_{0,1}, \dots, w_{0, n}))$ solves the LWS problem, since $T[0] = 0$ and thus the values of $t_k, k \in [n]$ are correctly initialized.

We solve $S$ using Algorithm~\vref{alg:staticlws}.
\begin{algorithm}
\caption{Reducing \LWS to \StaticLWS}
\label{alg:staticlws}
\begin{algorithmic}[1]
\Function{$S$}{$\{i,\dots,j\}, (t_i,\dots, t_j)$}
   \If{$i=j$}
     \State \Return $T[i] \gets t_i$
   \EndIf
   \State $m \gets \lceil \frac{j-i}{2} \rceil$
   \State $(T[i], \dots, T[i+m-1]) \gets S(\{i,\dots,i+m-1\}, (t_i, \dots, t_{i+m-1}))$ \label{line:firstcall}
   \State solve $\StaticLWS$ on the subinstance given by $I:= \{i, \dots, i+m-1\}$ and $J:=\{i+m, \dots, i + 2m - 1\}$:
   \State \Comment{obtains values $T'[k] = \min_{i\le k' < i + m} T[k'] + w_{k',k}$ for $k=i+m,\dots,i + 2m  - 1$.}
   \State $t'_k \gets \min\{t_k, T'[k]\}$ for all $k = i+m, \dots, i + 2m - 1$.
   \State $(T[i+m],\dots,T[i+2m - 1]) \gets S(\{i+m,\dots,i + 2m - 1\}, (t'_{i+m}, \dots, t'_{i+2m- 1}))$ \label{line:secondcall}
   \If{$j = i + 2m$}
	\State $T[j] := \min \{t_j, \min_{ i \le k < j} T[k] + w_{k,j}\}$.
   \EndIf
   \State \Return $(T[i], \dots, T[j])$
\EndFunction
\end{algorithmic}
\end{algorithm}
We briefly argue correctness, using the invariant that $t_k = \min_{0\le k' < i} T[k'] + w_{k',k}$ in every call to $S$. If $S$ is called with $i=j$, then the invariant yields $t_i = \min_{0 \le k' < i} T[k'] + w_{k',i} = T[i]$, thus $T[i]$ is computed correctly. For the call in Line~\ref{line:firstcall}, the invariant is fulfilled by assumption, hence the values $(T[i],\dots,T[i+m-1])$ are correctly computed. For the call in Line~\ref{line:secondcall}, we note that for $k=i+m,\dots,i+2m-1$, we have  \[t'_k = \min \{ t_k, T'[k] \} = \min \{ \min_{0\le k' < i} T[k'] + w_{k',k}, \min_{i \le k' < i+ m} T[k'] + w_{k',k}\} = \min_{0\le k' < i+ m} T[k'] + w_{k',k}.\] 
Hence the invariant remains satisfied. Thus, the values $(T[i+m],\dots,T[i+2m-1])$ are correctly computed. Finally, if $j=i+2m$, we compute the remaining value $T[j]$ correctly, since $t_j = \min_{0\le k < i} T[k] + w_{k,j}$ by assumption.

To analyze the running time $T^S(n)$ of $S$ on an interval of length $n:=j-i+1$, note that each call results in two recursive calls of interval lengths at most $n/2$. In each call, we need an additional overhead that is linear in $n$ and $T^\StaticLWS(n/2)$. Solving the corresponding recursion $T^S(n) \le 2T^S(n/2) + T^\StaticLWS(n/2) + \Oh(n)$, we obtain that an $\Oh(N^{2-\varepsilon})$-time algorithm \StaticLWS, with $0<\varepsilon<1$ yields $T^\LWS(n) \le  T^S(n) = \Oh(n^{2-\varepsilon})$. Similarly, an $\Oh(N \log^c N)$-time algorithm for \StaticLWS would result in an $\Oh(n \log^{c+1} n)$-time algorithm for \LWS.
\end{proof}

For the special case of $\LRLWS$, it is straightforward to see that the static version boils down to the following natural reformulation.

\begin{problem}[\AllInnProd]
Given $a_1,\dots, a_n \in \{-W,\dots,W\}^d$ and $b_1, \dots, b_n \in \{-W,\dots,W\}^d$, determine \emph{for all $j\in [n]$}, the value $\min_{i\in[n]} \langle a_i, b_j \rangle$. (Again, we typically assume that $d=n^{o(1)}$ and $W= 2^{n^{o(1)}}$.)
\end{problem}

\begin{lemma}[$\StaticLWS(\weightmat_\lowrank) \lequad \AllInnProd$]\label{lem:staticLRLWS} We have 
\[T^{\StaticLWS(\weightmat_\lowrank)}(n,d,W) \le T^{\AllInnProd}(n,d+1,nW) + \Oh(nd).\]
\end{lemma}
\begin{proof} Consider $\StaticLWS(\weightmat_\lowrank)$. Let $I=\{a+1, \dots, a+N\}$, $J=\{a+N+1, \dots, a+2N\}$ and values $T[a+1],\dots,T[a+N]$ be given. To determine $T'[j] = \min_{i \in I} T[i] + w_{i,j}$ for all $j\in J$, it is sufficient to solve \AllInnProd on the vectors $a_{a+1},\dots,a_{a+N}, b_{a+N+1},\dots,b_{a+2N} \in \{nW, \dots, nW\}^{d+1}$ defined by 
\begin{align*}
a_i & := (\mu_i, T[i]) & b_j & = (\sigma_{j}, 1), & & \text{for all } i\in I,j\in J,
\end{align*}
since then $\langle a_i, b_j \rangle = T[i] + \langle \mu_i, \sigma_{j} \rangle = T[i] + w_{i,j}$. The claim immediately follows (note that $|T[i]| \le nW$).
\end{proof}

Finally, inspired by an elegant trick of \cite{VassilevskaWW10}, we reduce $\AllInnProd$ to $\MinInnProd$. 

\begin{lemma}[$\AllInnProd \lequad \MinInnProd$]\label{lem:allToMinInn} We have 
\[T^\AllInnProd(n,d,W) \le \Oh(n \cdot T^\MinInnProd(\sqrt{n}, d+3, nd W^2) \cdot \log^2 nW).\]
\end{lemma}
\begin{proof}
We first observe that we can tune \MinInnProd to also return a witness $(i,j)$ with $\langle a_i, b_j \rangle \le r$, if it exists. To do so, we replace each $a_i$ by the $(d+2)$-dimensional vector $a_i' = (a_i \cdot n, (i-1) n, -1)$ and similarly, each $b_j$ by the $(d+2)$-dimensional vector $b_j' = (b_j \cdot n, -1, j-1)$. Clearly, we have $\langle a_i', b_j'\rangle = \langle a_i,b_j \rangle n^2 - (i-1) n - (j-1)$. Thus $\langle a_i',b_j' \rangle \le rn^2$ if and only if $\langle a_i, b_j \rangle \le r$ since $i,j\in [n]$. Using a binary search over $r$, we can find $\min_{i,j} \langle a_i', b_j' \rangle$, from whose precise value we can determine also a witness, if it exists. Thus the running time $\wit(n,d,W)$ for finding such a witness is bounded by $\Oh(\log nW) \cdot T^\MinInnProd(n,d+2,nW)$.

To solve $\AllInnProd$, i.e., to compute $p_j := \min_{i\in [n]} \langle a_i, b_j \rangle$ for all $j\in [n]$, we employ a parallel binary search. Consider in particular the following problem $\probm$: Given arbitrary $r_1,\dots, r_n$, determine for all $j \in [n]$ whether there exists $i\in [n]$ such that $\langle a_i, b_j \rangle \le r_j$. We will show below that this problem can be solved in time $\Oh(n \cdot \wit(\sqrt{n},d+1,dW^2))$. The claim then follows, since starting from feasible intervals $\intR_1 = \cdots = \intR_n = \{-dW^2,\dots,dW^2\}$ satisfying  $p_j \in \intR_j$, we can halve the sizes of each interval simultaneously by a single call to $\probm$. Thus, after $\Oh(\log (dW))$ calls, the true values $p_j$ can be determined, resulting in the time guarantee $T^\AllInnProd(n,d,w) = \Oh(n \cdot \wit(\sqrt{n}, d+1, dW^2) \cdot \log (dW)) = \Oh(n \cdot T^\MinInnProd(\sqrt{n}, d+3, n dW^2) \log^2(nW))$, as desired.

We complete the proof of the claim by showing how to solve $\probm$. Without loss of generality, we can assume that $r_j \le dW^2$ for every $j$, since no larger inner product may exist. We group the vectors $a_1,\dots,a_n$ in $g := \lceil \sqrt{n} \rceil$ groups $A_1,\dots,A_g$ of size at most $\sqrt{n}$ each, and do the same for the vectors $b_1,\dots,b_n$ to obtain $B_1,\dots,B_g$. Now, we iterate over all pairs of groups $A_k, B_\ell$, $k,\ell \in [g]$: For each such choice of pairs, we do the following process. For each vector $a_i \in A_k$, we define the $(d+1)$-dimensional vector $\tilde{a}_i := (a_i, -1)$ and for every vector $b_j\in B_\ell$, we define $\tilde{b}_j:= (b_j, r_j)$. In the obtained instance $\{ \tilde{a}_i \}_{a\in A_k}, \{ \tilde{b}_j \}_{b\in B_\ell}$, we try to find some $i,j$ such that $\langle \tilde{a}_i, \tilde{b}_j \rangle \le 0$, which is equivalent to $\langle a_i, b_j \rangle \le r_j$. If we succeed in finding such a witness, we delete $b_j$ and $\tilde{b_j}$ (but remember its witness) and repeat finding witnesses (an deleting the witnessed $b_j$) until we cannot find any. The process then ends and we turn to the next pair of groups.

It is easy to see that for all $j\in [n]$, we have $\langle a_i, b_j \rangle \le r_j$ for some $i\in[n]$ if and only if the above process finds a witness for $b_j$ at some point. To argue about the running time, we charge the running time of every call to witness finding to either (1) the pair $A_k,B_\ell$, if the call is the first call in the process for $A_k, B_\ell$, or (2) to $b_j$, if the call resulted from finding a witness for $b_j$ in the previous call. Note that every pair $A_k, B_\ell$ is charged by exactly one call and every $b_j$ is charged by at most one call (since in after a witness for $b_j$ is found, we delete $b_j$ and no further witness for $b_j$ can be found). Thus in total, we obtain a running time of at most $(g^2 + n) \cdot  \wit(\sqrt{n}, d+1, dW^2) + \Oh(n) = \Oh(n \cdot \wit(\sqrt{n}, d+1, dW^2))$.
\end{proof}

\begin{theorem}\label{thm:lrlws}
We have $\LRLWS \eqquad \MinInnProd$.
\end{theorem}
\begin{proof}
In Lemmas~\ref{lem:LRLWS-lb},~\ref{lem:dynToStatic},~\ref{lem:staticLRLWS}, and~\ref{lem:allToMinInn}, we have proven
\begin{multline*}
 \MinInnProd \lequad \LRLWS = \LWS(\weightmat_\lowrank) \\
 \lequad \StaticLWS(\weightmat_\lowrank) \lequad \AllInnProd \lequad \MinInnProd, 
\end{multline*}
proving the claim.
\end{proof}

\section{Coin Change and Knapsack Problems}
\label{sec:coinchange}

In this section, we focus on the following problem related to \Knapsack: Assume we are given coins of denominations $d_1, \dots, d_m$ with corresponding weights $w_1,\dots,w_m$ and a target value $n$, determine a way to represent $n$ using these coins (where each coin can be used arbitrarily often) minimizing the total sum of weights of the coins used. Since without loss of generality $d_i \le n$ for all $i$, we can assume that $m\le n$ and think of $n$ as our problem size. In particular, we describe the input by weights $w_1, \dots, w_n$ where $w_i$ denotes the weight of the coin of denomination $i$ (if no coin with denomination $i$ exists, we set $w_i = \infty$). It is straightforward to see that this problem is an $\LWS$ instance $\LWS(\weightmat_\coinC)$, where the weight matrix $\weightmat_\coinC$ is a Toeplitz matrix.

\begin{problem}[\CoinC] We define the following \LWS instantiation $\CoinC = \LWS(\weightmat_\coinC)$.\\
\Data{weight sequence $w = (w_1, \dots, w_n)$ with $w_i \in \{-W, \dots, W\} \cup \{\infty\}$}
\Weights{$w_{i,j} = w_{j-i}$ for $0\le i < j \le n$}
\end{problem}

Translated into a Knapsack-type formulation (i.e., denominations are weights, weights are profits, and the objective becomes to maximize the profit), the problem differs from \UnboundedKnapsack only in that it searches for the most profitable multiset of items of weight \emph{exactly} $n$, instead of \emph{at most} $n$.

\UKnapsackproblem

The purpose of this section is to show that both \CoinC and \UnboundedKnapsack are subquadratically equivalent to the \mpconv problem. Along the way, we also prove quadratic-time \mpconv-hardness of \Knapsack. Recall the definition of \mpconv.

\mpconvproblem

As opposed to the classical convolution, which we denote as $a\classicconv b$, solvable in time $\Oh(n\log n)$ using FFT, no strongly subquadratic algorithm for \mpconv is known. Compared to the popular orthogonal vectors problem, we have less support for believing that no $\Oh(n^{2-\varepsilon})$-time algorithm for \mpconv exists. In particular, interesting special cases can be solved in subquadratic time~\cite{ChanL15} and there are subquadratic-time co-nondeterministic and nondeterministic algorithms~\cite{BremnerCDEHILPT14, CarmosinoGIMPS16}. At the same time, breaking this long-standing quadratic-time barrier is a prerequisite for progress on refuting the 3SUM and APSP conjectures. This makes it an interesting target particularly for proving subquadratic \emph{equivalences}, since both positive and negative resolutions of this open question appear to be reasonable possibilities.

To obtain our result, we address two issues: (1) We show an equivalence between the problem of determining only the value $T[n]$, i.e., the best way to give change only for the target value $n$, and to determine \emph{all values} $T[1], \dots, T[n]$, which we call the \emph{output-intensive version}. (2) We show that the output-intensive version is subquadratic equivalent to \mpconv.

\begin{problem}[\oiCoinC] 
The output-intensive version of \CoinC is to determine, given an input to $\CoinC$, all values $T[1],\dots,T[n]$.
\end{problem}

We first consider issue (2) and provide a \mpconv-based lower bound for \oiCoinC.

\begin{lemma}[$\mpconvmath \lequad \oiCoinC$]
\label{lem:convTooiCC}
We have $T^\mpconvmath(n, W) \le T^\oiCoinC(6n, 4(2W+1)) + \Oh(n)$.
\end{lemma}
\begin{proof}
We first do a translation of the input. Note that for any scalars $\alpha, \beta$, we have $(a+\alpha) \conv (b+\beta) = (a \conv b) +  \alpha + \beta$. Let $M:= 2W +1$. Without loss of generality, we may assume that 
\begin{align*}
2M  & \le  a_i \le 3M & &\text{for all } i=0, \dots, n-1, \\
0 & \le  b_j \le M & &\text{for all } j=0, \dots, n-1.
\end{align*}
We now define a \CoinC instance with a problem size $n'=6n$ and $W' = 4M$ by defining
\[ w = (4M)^n \concat (a_{n-1}, \dots, a_0) \concat (4M)^n \concat (b_{n-1}, \dots, b_0) \concat (4M)^{2n}. \]

We now claim that $T[4n + i] = (a\conv b)_{2n-i}$ for $i=1, \dots, 2n$, which immediately yields the lemma. To do so, we will prove the following sequence of identities. 
\begin{align}
T[i]  &= 4M & &\text{for } i \in [n], \label{eq:range1}\\
T[n+i]  &= a_{n-i} & &\text{for } i \in [n], \label{eq:range2}\\
T[2n+i]  &= 4M & &\text{for } i \in [n], \label{eq:range3}\\
T[3n+i]  &= b_{n-i} & &\text{for } i \in [n], \label{eq:range4}\\
T[4n+i]  &= (a\conv b)_{2n-i} & &\text{for } i \in [2n], \label{eq:range5}
\end{align}
In the last line, we define, for our convenience, $(a\conv b)_{2n-1} = 4M$ (note that before, we defined only the entries $(a\conv b)_k$ with $k \le 2n-2$).

For later convenience, observe that $0 \le w_i \le 4M$ for all $i \in [n']$. It is easy to see that this implies $0 \le T[i] \le 4M$ for $i\in [n']$.

The identities in~\eqref{eq:range1} are obvious.

To prove the identities in~\eqref{eq:range2} inductively over $i$, recall that $T[n+i] = \min_{j=1, \dots, n+i}\{T[n+i-j] + w_j\}$. Observe that $T[n+i-j] + w_j < 4M$ can only occur if $j \ge n+1$ (since otherwise $w_j = 4M$), which implies $n+i-j \le n$ and $T[n+i-j] = 4M$ except for the case $j=n+i$. In this case, we have $T[n+i-j] + w_j = T[0] + w_{n+i} = a_{n-i} \le 4M$.

To prove the identities in~\eqref{eq:range3}, observe that for $1 \le j \le 3n$, we have $w_j \ge 2M$ by assumption $\min_i a_i \ge 2M$. Similarly, we have already argued that $T[i'] \ge 2M$ for $1\le i' \le 2n$. Thus, we can inductively show that $T[2n+i] = \min \{T[0] + w_{2n+i}, \min_{j=1, \dots, 2n+i-1} T[2n+i-j] + w_j \} = 4M$ using $w_{2n+i} = 4M$ and that every sum in the inner minimum expression is at least $4M$.

To prove the identities in~\eqref{eq:range4}, note that for $T[3n+i-j] + w_j < 4M$ to hold, we must have either $n + 1 \le j \le 2n$ or $3n + 1 \le j \le 3n+i$, since otherwise $w_j = 4M$. We observe that for $n+1 \le j \le 2n$, we have $w_j \ge \min_i a_i \ge 2M$ and $T[3n+i - j]\ge \min_i a_i \ge 2M$. Thus, we may assume that $3n+1 \le j \le 3n+i$. Note that in this case, we have $T[3n+i - j] = 4M$ except for the case $j=3n+i$, where we have $T[3n+i -j] + w_j = T[0] + w_{3n+i} = b_{n-i} < 4M$.

Finally, for the identities in~\eqref{eq:range5}, we might have $T[4n+i] + w_j < 4M$ only if $n+1 \le j \le 2n$ or $3n+1 \le j \le 4n$. First consider the case that $i=1$.  We have 
\[ T[4n+1] = \min \{ w_{4n+1}, \min_{n+1 \le j \le 2n} \underbrace{T[4n+1 - j]}_{ = 4M} + w_j, \min_{3n+1 \le j \le 4n} \underbrace{T[4n+1 - j]}_{= 4M} + w_j \} = 4M.\]
Inductively over $1 < i \le 2n$, we will prove $T[4n+i] = (a\conv b)_{2n-i}$. By definition,
\begin{align}
T[4n+i] & = \min \{ w_{4n+i}, \min_{n+1 \le j \le 2n} T[4n+i-j] + w_j, \min_{3n+1 \le j \le 4n} T[4n+i - j] + w_j \} \nonumber\\
& = \min \{ w_{4n+i}, \min_{1 \le j' \le n} T[3n+i-j'] + a_{n-j'}, \min_{1 \le j' \le n} T[n+i - j'] + b_{n-j'}\} \label{eq:convlastpart}
\end{align}
Note that
\[ \min_{1 \le j' \le n} \underbrace{T[n+i - j']}_{= 4M \text{ for } j' \ge i \text{ or } j' < i-n} + b_{n-j'} = \min_{\max\{1,i-n\} \le j' \le \min\{i-1,n\}} a_{n-(i-j')} + b_{n-j'} = (a\conv b)_{2n-i} \]
where the last equation follows from noting that the choice of $j'$ lets $n-j'$ and $n-(i-j')$ range over all admissible pairs of values in $\{0, \dots, n-1\}$ summing up to $2n-i$. Similarly, we inductively prove that 
\[ \min_{1 \le j' \le n} T[3n+i - j'] + a_{n-j'} = \min_{\max\{1,i-n\} \le j' \le \min\{i-1,n\}} a_{n-(i-j')} + b_{n-j'} = (a\conv b)_{2n-i}, \]
since $a_{n-j'} \ge 2M$ and $T[3n+i-j'] \ge 2M$ whenever $j' \ge i$ or $j' < i-n$ (where the last regime uses $T[4n+i'] = (a\conv b)_{2n-i'} \ge 2M$ inductively for $i' < i$). Finally, since $(a \conv b)_{2n-i} \le (\max_i a_i) + (\max_j b_j) \le 4M$, we can simplify \eqref{eq:convlastpart} to $T[4n+i] = (a\conv b)_{2n-i}$.
\end{proof}

Using the notion of \StaticLWS, the other direction is straight-forward.

\begin{lemma}\label{lem:oiCoinCtompconv}
We have $\oiCoinC \lequad \StaticLWS(\weightmat_\coinC) \lequad \mpconvmath$.
\end{lemma}
\begin{proof}
In Lemma~\ref{lem:dynToStatic}, we have in fact reduced the
output-intensive version of $\LWS(\weightmat)$ to our static
problem $\StaticLWS(\weightmat)$, thus specialized to the coin change problem, we only need to show that $\StaticLWS(\weightmat_\coinC)$ subquadratically reduces to \mpconv. Consider an input instance to $\StaticLWS$ given by $I= \{a+1, \dots, a+N\}$, $J=\{a+N+1, \dots, a+2N\}$ and values $T[i], i \in I$. Defining $M:=2W+1$ and the vectors 
\begin{align*}
u &:= (nM, T[a+1],\dots,T[a+N], \overbrace{nM, \dots, nM}^{N \text{ times}} ),  \\
v &:= (nM, w_1, \dots, w_{2N}),
\end{align*}
we have $(u \conv v)_{N+k} = \min_{i=1, \dots, N} T[a+i] + w_{N+k-i} = T'[a+N+k]$ for all $k = 1, \dots, N$, thus a \mpconv of two $(2n+1)$-dimensional vectors solves $\StaticLWS(\weightmat_\coinC)$, yielding the claim.
\end{proof}

The last two lemmas resolve issue (2). We proceed to issue (1) and show that the output-intensive version is subquadratically equivalent to both \CoinC and \UnboundedKnapsack that only ask to determine a single output number. We introduce the following notation for our convenience: Recall that weight $w_i$ denotes the weight of a coin of denomination $i$. For a multiset $S\subseteq [n]$, we let $d(S) := \sum_{i \in S} i$ denote its \emph{total denomination}, i.e., sum of the denomination of the coins in $S$ (where multiples uses of the same coin is allowed, since $S$ is a multiset). We let $w(S):= \sum_{i\in S} w_i$ denote the weight of the multiset. Analogously, when considering a Knapsack instance, $p(S) = \sum_{i} p_i$ denotes the total profit of the item (multi)set $S$.

It is trivial to see that $\UnboundedKnapsack \lequad \oiCoinC$. Furthermore, we can give the following simple reduction from \CoinC to \UnboundedKnapsack.
\begin{obs}[$\CoinC \lequad \UnboundedKnapsack \lequad \oiCoinC$]\label{obs:Uknap}
We have \\$T^\CoinC(n, W) \le T^\UnboundedKnapsack(n, nW)+\Oh(n)$ and $T^\UnboundedKnapsack(n, W) \le T^\oiCoinC(n,W)+\Oh(n)$.
\end{obs}
\begin{proof}
Given a \CoinC instance, for every weight $w_i<\infty$, we create an item of size $i$ and profit $p_i := i \cdot M - w_i$ in our resulting \UnboundedKnapsack instance for a sufficiently large constant $M\ge nW$. This way, all profits are positive and every multiset $S$ whose sizes sum up to $B$ has a profit of $p(S) = B \cdot M - w(S)$. Since $M \ge nW \ge \max_{S, d(S) \le n} |w(S)|$, this ensures that the maximum-profit multiset of total size/denomination at most $n$ has a total size/denomination of exactly $n$. Thus, the optimal multiset $S^\ast$ has profit $p(s^\ast) = n \cdot M - \min_{S: d(S) = n} w(S) = n\cdot M -T[n]$, from which we can derive $T[n]$, as desired.

Given an \UnboundedKnapsack instance, we define for every item of size $i$ and profit $p_i$ the corresponding weight $w_i = - p_i$ in a corresponding \CoinC instance. It remains to compute all $T[1], \dots, T[n]$ in this instance and determining their minimum, concluding the reduction.
\end{proof}

The remaining part is similar in spirit to Lemma~\ref{lem:allToMinInn}: Somewhat surprisingly, the same general approach works despite the much more sequential nature of the Knapsack/CoinChange problem -- this sequentiality can be taken care of by a more careful treatment of appropriate subproblems that involves solving them in a particular order and feeding them with information gained during the process.

In what follows, to clarify which instance is currently considered, we let $T^\instI$ denote the $T$-table of the \textsc{(oi)CC} LWS problem (see Problem~\ref{prob:lws}) corresponding to instance $\instI$. Dropping the superscript always refers to $T^\instI$.

\begin{lemma}[$\oiCoinC \lequad \CoinC$]\label{lem:oiCCtoCC}
We have that $T^\oiCoinC(n,W) \le \Oh(\log(nW) \cdot n \cdot T^\CoinC(24\sqrt{n}, 3n^2W))$.
\end{lemma}
\begin{proof}
Let $\instI$ be an \oiCoinC instance. 
To define our subproblems, we set $N:= \lceil \sqrt{n} \rceil$  and define $N$ ranges $\rangeW_1 := \{1, \dots, N\}$, $\dots$, $\rangeW_N := \{(N-1)N + 1, \dots, N^2\}$. To determine all $T[i] = \min_{S: d(S) = i} w(S)$, we will compute $T[i]$ for all $i \in \rangeW_j$ successively over all $j =1, \dots, N$. The case of $j=1$ and $j=2$ can be computed by the naive algorithm in time $\Oh(N^2) = \Oh(n)$. Consider now any fixed $j\ge 3$ and assume that all values $T[i]$ for $i\in \rangeW_{j'}$ with $j' < j$ have already been computed. We employ a parallel binary search. For every $i \in \rangeW_j$, we set up a feasible range $\range_i$ initialized to $\{-nW,\dots,nW\}$. We will maintain the invariant that $T[i] \in \range_i$ and will halve the size of all feasible ranges $\range_i, i \in \rangeW_j$ simultaneously using a small number of calls to the following problem $\prob(M,\bar{W})$: Given an instance~$\instJ$ for \CoinC specified by the weights $\tilde{w}_1,\dots,\tilde{w}_M$, as well as values $\tilde{r}_1,\dots, \tilde{r}_M \in \{-\bar{W}, \dots, \bar{W}\} \cup \{-\infty, \infty\}$, determine whether there exists an $i\in [M]$ with $T^\instJ[i] \le \tilde{r}_i$, and if so, also return a witness $i$. We will later prove that this problem can be solved in time $T^{\prob}(M, \bar{W}) = \Oh(T^\CoinC(2M,3M^2\bar{W}))$. Clearly, after $\Oh(\log(nW))$ rounds of this parallel binary search, the feasible ranges consists of single values, thus determining the values of all $T[i]$ for $i\in \rangeW_j$. Since we will show that halving all feasible ranges for range $\rangeW_j$ takes $\Oh(N)$ calls to $\prob(12N, nW)$, and we need to determine at most $N$ ranges $\rangeW_3, \dots, \rangeW_N$, the total time for this process amounts to $\Oh(\log(nW) N^2 \cdot T^{\prob}(12N, nW)) = \Oh(\log(nW) N^2 \cdot  T^\CoinC(24N, 3n^2W))$.

We now describe how to use $\prob$ to halve the size of all feasible ranges $\range_i, i \in \rangeW_j$: we set $r_i$ to the median of $\range_i$ and aim to determine, for all $i\in \rangeW_j$, whether $T[i]\le r_i$, i.e., whether some multiset $S$ with $d(S) = i$ and $w(S) \le r_i$ exists. We achieve this by the following process: For every $k=1, \dots, j$, we consider only two ranges, namely $\rangeW_k = \{(k-1)N+1, \dots, kN\}$ and $\rangeW_{j-k} \cup \rangeW_{j-k+1} = \{(j-k-1)N+1,\dots,(j-k+1)N\}$.
Let us first consider the case $k\ge 2$. Here, we can define the $2N$-dimensional vectors $a,b$ with
\begin{align*}
a_\ell & = \begin{cases} w_{(k-1)N + \ell} & \text{for } \ell \in [N], \\ \infty & \text{for } \ell > N, \end{cases}\\
b_\ell & = T[(j-k-1)N + \ell] \qquad \text{for } \ell \in [2N].
\end{align*} 
(Note that all $T[i], i \in \rangeW_{j-k} \cup \rangeW_{j-k+1}$ for $k\ge 2$ have already been computed by assumption.)
We are interested in all those values of the \mpconv $a \conv b$ of these vectors that correspond to summing up some $w_{(k-1)N + \ell}$ with some $T[(j-k-1)N+\ell']$ such that $(j-2)N + \ell + \ell' \in \rangeW_j$. More specifically, we aim to determine whether there is some $\ell$ with $(a \conv b)_{N+\ell} \le r_{(j-1)N+\ell}$. To do so, we use the reduction from \mpconv to \oiCoinC given in Lemma~\ref{lem:convTooiCC} to create an \oiCoinC instance $\instJ$. From this instance of problem size $12N$ we can read off the values of $a\conv b$ as a certain interval in the corresponding $T^\instJ$-table. Thus, we can test whether $(a \conv b)_{N+\ell} \le r_{(j-1)N+\ell}$ for some $\ell$ using $\prob(12N, nW)$: for every $\ell$, we let $i$ be the unique index in the $T^\instJ$-table representing the entry $(a\conv b)_{N+\ell}$ and set $\tilde{r}_i := r_{(j-1)N+\ell}$. For all other $i'$, we set $\tilde{r}_{i'} = -\infty$, thus enforcing that those indices will never be reported.

For the special case $k=1$, we proceed slightly differently: Here, we define the $2N$-dimensional vectors $a,b$ with
\begin{align*}
a_\ell & = T[\ell] \qquad \text{for } \ell \in [2N]\\
b_\ell & = \begin{cases} T[(j-2)N + \ell] & \text{for } \ell \in [N]\\ \infty & \text{for } \ell > N. \end{cases}
\end{align*} 
(Note that all necessary $T[i], i\in \rangeW_1\cup\rangeW_2$ and $T[i],i\in \rangeW_{j-1}$ have already been computed by assumption.) 
Analogously to above, we use $\prob(12N, nW)$ to test whether $(a\conv b)_{N+\ell} \le r_{(j-1)N+\ell}$ using the reduction from \mpconv to \oiCoinC given in Lemma~\ref{lem:convTooiCC}.

Once an $i\in \rangeW_j$ has been reported to satisfy $T[i] \le r_i$ for some witnessing subproblem given by the ranges $\rangeW_k$ and $\rangeW_{j-k} \cup \rangeW_{j-k+1}$ for some $k$, we set $r_i := -\infty$ and repeat on the same subproblem $k$ (analogously to the approach of Lemma~\ref{lem:allToMinInn}). Note that for every $j$, we have $j\le N$ subproblems and at most $N$ many indices $i\in \rangeW_j$ that can be reported. Thus, we use at most $\Oh(N)$ many calls to the subproblem $\prob$.

To briefly argue correctness, note that by construction, we only determine some $i$ with $T[i] \le r_i$ if we have found a witness. For the converse, let $k$ be the largest index such that the optimal multiset for $i$ includes a coin in $\rangeW_k$. Then the subproblem given by the ranges $\rangeW_k$ and $\rangeW_{j-k} \cup \rangeW_{j-k+1}$ will give a witness. This is obvious for $k\ge 2$. For $k=1$, note that no weight in $\rangeW_{k'}$ with $k'>1$ is used in an optimal multiset for $T[i]\in \rangeW_j$. In particular, the optimal multiset $S$ can be represented as $S= S' \cup S''$, where $S'$ is a multiset of total denomination $i'\in \rangeW_{j-1}$ and $S''$ is a multiset of total denomination $i-i' \in \rangeW_{1}\cup \rangeW_{2}$. Thus, in the instance constructed from $a,b$, we will find the witness $T[i] \le T[i'] + T[i-i'] \le r_i$.

We finally describe how to solve $\prob(M, \bar{W})$ in time $T^\CoinC(2M, 3M^2\bar{W})$. First consider the problem \emph{without finding a witnessing $i$}. Let $\tilde{w}_1, \dots, \tilde{w}_M, \tilde{r}_1, \dots, \tilde{r}_M$ be an instance $\instJ$ of $\prob(M, \bar{W})$. We  define a \CoinC instance $\instK$ of problem size $2M$ by giving the weights
\begin{align*}
w'_i & := \tilde{w}_i & \text{for all } i\in [M],\\
w'_{2M-i} & := -3M\bar{W} - \tilde{r}_i & \text{for all } i \in [M].
\end{align*}
We claim that $T^\instK[2M] \le -3M\bar{W}$ iff the input instance to $\prob$ is a yes instance: First observe that $T^\instK[1]=T^\instJ[1], \dots, T^\instK[M] =T^\instJ[M]$ since the first $M$ weights agree for both $\instJ$ and $\instK$. Consider the case that there is some $i\in [M]$ with $T^\instJ[i] \le \tilde{r}_i$. Then we have $T^\instK[2M] \le T^\instK[i] + w_{2M-i} = (T^\instJ[i] - \tilde{r}_i) - 3M\bar{W} \le -3M\bar{W}$, as desired. Conversely, assume that all $T^\instJ[i] > \tilde{r}_i$. We distinguish the cases whether the optimal subsequence $S$ uses only weights among $\tilde{w}_1, \dots, \tilde{w}_M$ or not. In the first case, since $|\tilde{w}_i| \le W$ for $i\in [M]$, we have that $w(S) \ge 2M \cdot \min_{i\in [n]} |\tilde{w}_i|\ge -2M\bar{W} > -3M\bar{W}$. Otherwise, $S$ uses exactly one weight among $\tilde{w}_{M+1}, \dots, \tilde{w}_{2M}$. Let this weight be $\tilde{w}_{2M-i}$. Then $w(S) = T^\instK[i] + \tilde{w}_{2M-i} = (T^\instJ[i] - \tilde{r}_i) - 3M\bar{W} > - 3M\bar{W}$ since $T^\instJ[i] > \tilde{r}_i$, yielding the claim. 

Very similar to Lemma~\ref{lem:allToMinInn}, we can now tune the above reduction to also produce a witness $i$ such that $T^\instJ[i] \le \tilde{r}_i$. For this, we scale all weights $w'_i, i\in [2M]$ by a factor of $M$ and subtract a value of $i-1$ for every $w'_i, i\in [M]$. It is easy to see that a yes instance $\instK$ attains some value $T^\instK[2M] = -\kappa \cdot M - i$ for some integers $\kappa \ge 3$ and $0\le i<n$, where $i+1$ is a witness for $T^\instJ[i+1] \le \tilde{r}_{i+1}$, thus computing $T^\instK[2M]$ lets us derive a witness as well. Thus, problem $\prob$ can be solved by a single call to $T^\CoinC(2M, 3M^2\bar{W})$.
\end{proof}

The results above prove the following theorem.
\begin{theorem}\label{thm:coinchange}
We have $\mpconvmath \eqquad \CoinC \eqquad \UKnapsack$. Furthermore, the bounded version of \Knapsack admits no strongly subquadratic-time algorithm unless \mpconv can be solved in strongly subquadratic time.
\end{theorem}
\begin{proof}
Lemmas~\ref{lem:convTooiCC} and ~\ref{lem:oiCoinCtompconv} prove $\mpconvmath \eqquad \oiCoinC$, while Observation~\ref{obs:Uknap} and Lemma~\ref{lem:oiCCtoCC} establish $\oiCoinC \eqquad \CoinC \eqquad \UnboundedKnapsack$, yielding the first claim.

The second claim follows from inspecting the proofs of Lemma~\ref{lem:convTooiCC}, Lemma~\ref{lem:oiCCtoCC} and the first claim of Observation~\ref{obs:Uknap} and observing that we only reduce to \CoinC/\Knapsack instances in which the optimal multiset (for each total size) is always a set, i.e., uses each element at most once.
\end{proof}

\section{Chain LWS}
\label{sec:chains}

In this section we consider a special case of of Least-Weight
Subsequence problems called the Chain Least-Weight Subsequence. This
captures problems in which edge weights are given implicitly by a
relation $\rel$ that determines which pairs of data items we are
allowed to chain -- the aim is to find the longest chain.

An example of a Chain Least-Weight Subsequence problem is the
\NestedBoxes problem. Given $n$ boxes in $d$ dimensions,
given as non-negative, $d$-dimensional vectors $b_1, \ldots, b_n$,
find the longest chain such that each box fits into the next (without
rotation). We say box that box $a$ fits into box $b$ if for all
dimensions $1 \leq i \leq d$, $a_i \leq b_i$.

\NestedBoxes is not immediately a least-weight subsequence problem, as
for least weight subsequence problems we are given a sequence of data
items, and require any sequence to start at the first item and end at
the last. 
We can easily convert \NestedBoxes into a \LWS problem by sorting the
vectors by the sum of the entries and introducing two special boxes,
one very small box $\bot$ such that $\bot$ fits into any box $b_i$ and
one very large box $\top$ such that any $b_i$ fits into $\top$.

We define the chain least-weight subsequence problem with respect to
any relation $\rel$ and consider a weighted version where data items
are given weights. To make the definition consistent with the
definition of $\LWS$ the output is the weight of the sequence that
minimizes the sum of the weights.

\begin{problem}[\CLWS]
Fix a set of objects $X$ and a relation $\rel \subseteq X\times X$. 
We define the following \LWS instantiation $\CLWS(\rel) = \LWS(\weightmat_{\CLWS(R)})$.\\
\Data{sequence of objects $x_0, \dots, x_n \in X$ with weights $w_1, \dots, w_n \in \{-W,\ldots,W\}$.} 
\Weights{$w_{i,j} = \begin{cases} w_j & \text{if } (x_i, x_j) \in \rel, \\ \infty & \text{otherwise},\end{cases}$ for $0\le i < j \le n$.}
\end{problem}

The input to the (weighted) chain least-weight subsequence problem is
a sequence of data items, and not a set. Finding the longest chain in
a set of data items is $\NP$-complete in general. For example,
consider the box overlap problem: The input is a set of boxes in two
dimensions, given by the top left corner and the bottom right corner,
and the relation consists of all pairs such that the two boxes
overlap. This problem is a generalization of the Hamiltonian path
problem on induced subgraphs of the two-dimensional grid, which is an
$\NP$-complete problem \cite{ItaiPS82}.

We relate $\CLWS(\rel)$ to the class of \emph{selection} problems with
respect to the same relation $\rel$. 

\begin{problem}[Selection Problem]
  Given data items $a_1, \ldots, a_n, b_1, \ldots, b_n$ and a relation
  $\rel(a_i, b_j)$, determine if there is a pair $i,j$ satisfying
  $\rel(a_i, b_j)$. We denote this selection problem with respect to a
  relation $\rel$ by $\Selection(\rel)$.
\end{problem}

The class of selection problems includes several well studied problems
including \MinInnProd, \OV \cite{Williams05, AbboudWY15} and
\VectorDomination \cite{ImpagliazzoLPS14}.

We will use the selection problems in the search variant, where we
find a pair satisfying the $\rel$ if such a pair exists. To reduce the
the search variant to the decision variants in a fine-grained way, we
can use a simple, binary search type reduction from the decision
problem to the search problem:

We give a subquadratic reduction from $\CLWS(\rel)$ to $\Selection(\rel)$ that is
independent of $\rel$. 

\begin{theorem}
  \label{thm:LWStoSel}
  For all relations $\rel$ such that $\rel$ can be computed in time
  subpolynomial in the number of data items $n$,
  $\CLWS(\rel) \le_2 \Selection(\rel)$.
\end{theorem}

The proof is again based on \StaticLWS and a variation on a trick
of~\cite{VassilevskaWW10}.

As an intermediate step, we define $\StaticCLWS$ as the equivalent
of $\StaticLWS$ in the special case for chains.

\begin{problem}[\StaticCLWS]
 Fix an instance of $\CLWS(\rel)$. Given intervals
 $I:=\{a+1, \dots, a+N\}$ and $J:=\{a+N+1,\dots, a+2N\}$ for some $a$
 and $N$, together
 with the correctly computed values $T[a+1],\dots,T[a+N]$, the
 Static Chain Least-Weight Subsequence Problem (\StaticCLWS) asks to
 determine
 \begin{align*}
   T'[j] & := \min_{i \in I: \rel(i,j)} T[i] + w_{j}  & \text{for all } j\in J. 
 \end{align*}
\end{problem}

Similar to the definition of $\CLWS$, $\StaticCLWS$ is the special case of
\StaticLWS where the the weights $w_{i,j}$ are restricted to be either
$w_j$ or $\infty$, depending on $\rel$. As a result, Lemma \ref{lem:dynToStatic} applies
directly.

\begin{corollary}[$\CLWS(\rel) \lequad\StaticLWS(\rel)$]
 \label{lem:chainDynToStatic}
 For any $\rel$, if $\StaticCLWS(\rel)$ can be solved in
 time $\Oh(n^{2-\varepsilon})$ for some $\varepsilon>0$, then
 $\CLWS(\rel)$ can be solved in time $\tOh(n^{2-\varepsilon})$.
\end{corollary}

We now reduce $\StaticCLWS(\rel)$ to $\Selection(\rel)$ with a variation on
the trick by \cite{VassilevskaWW10}.

\begin{lemma}[$\StaticCLWS(\rel) \lequad \Selection(\rel)$]
 For all relations $\rel$ such that $\rel$ can be computed in time
 subpolynomial in the number of data items $n$, $\StaticCLWS(\rel)
 \lequad \Selection(\rel)$.
\end{lemma}
\begin{proof}
 As a first step, we sort the data items $a_i, i\in I= \{a+1, \dots, a+N\}$ by
 $T[i]$ in increasing order and we will assume for the remainder of
 the proof that for all $a+1 \leq i < a+N$ we have
 $T[i] \leq T[i+1]$. We then split the set $a_{a+1},\ldots, a_{a+N}$
 into $g := \lceil \sqrt{N} \rceil$ groups $A_1, \ldots, A_g$ with
 $A_i = \{a_{(i-1)\lceil N/g \rceil}, \ldots, a_{i\lceil N/g \rceil -
   1}\}$.  We split the set $b_{a+N+1},\dots,b_{a+2N}$ into $B_1, \ldots, B_g$ in a similar
 fashion. We then iterate over all pairs $A_k, B_l$ with
 $k,l \in [g]$ in lexicographic order, and for each pair we do the
 following. Call the oracle for $\Selection(\rel)$ on the input
 $A_k, B_l$ to find a pair $a_i, b_j$ such that the relation $\rel$
 is satisfied on the pair. If there is no such pair, move to the next
 pair $A_{k^*}, B_{l^*}$ of sets of data items. If there is such a pair,
 find the first element $a_{i^*} \in A_k$ such that $R(a_{i^*}, b_j)$
 using a simple linear scan. As we first sorted $A$ and iterate over
 sets $A_k, B_l$ in lexicographic order, we have
 $T'[j] = T[i^*] + w_j$. We then remove $b_j$ from $B_l$ and repeat.

 For the runtime analysis, we observe, that the oracle can find a
 pair of elements at most $\Oh(N)$ times, as each time we find a pair
 we remove an element from the input. In the case where we do find a
 pair of elements we do a linear scan that takes $\Oh(N/g)$
 time. Furthermore, each pair of sets $A_k, B_l$ can fail to find a
 pair at most once. Hence, if $T^{\Selection}$ is the time to solve
 the selection problem and using $g = \sqrt{N}$ we get a time of
 \begin{equation}
   T(N) = N T^{\Selection}(\sqrt{N}) + N (T^{\Selection}(\sqrt{N}) + \sqrt{N})
   = N T^{\Selection}(\sqrt{N})
 \end{equation}
 which is subquadratic if $T^{\Selection}(N)$ is subquadratic.
\end{proof}

\begin{theorem}
  \label{thm:SelToLWS}
  Let $D$ be the set of possible data items. For any relation $\rel$ such that
  \begin{itemize}
  \item There is a data item $\bot$ such that $(\bot, d) \in \rel$ for all $d \in D$.
  \item There is a data item $\top$ such that $(d, \top) \in \rel$ for all $d \in D$.
  \item For any set of data items $d_1, \ldots, d_n$ there is a
    sequence $i_1, \ldots, i_n$ such that for any $j < k$,
    $(d_{i_j}, d_{i_k}) \not\in \rel$. This ordering can be computed in
    time $\Oh(n^{2-\delta})$ for $\delta > 0$. We call this ordering the
    natural ordering. 
  \end{itemize}
  Then $\Selection(\rel) \lequad \CLWS(\rel)$. 
\end{theorem}
\begin{proof}
  We construct an unweighted $\CLWS$ problem with all weights set to
  $-1$, so that the problem is to find the longest chain. Let
  $a_1, \ldots a_n$ and $b_1, \ldots, b_n$ be the data items of
  $\Selection(\rel)$ and sort both sets according to the natural
  ordering. We claim that for the sequence of data items
  $\bot, a_1, \ldots a_n, b_1, \ldots, b_n, \top$ the weight of the
  least weight subsequence is $-3$ exactly if there is a pair
  $(a_i, b_j) \in \rel$. Because of the property of the natural
  ordering, any valid subsequence starting at $\bot$ and ending at
  $\top$ contains at most one element $a_i$ and at most one element
  $b_j$. If there is a pair $(a_i, b_j) \in \rel$, then the sequence
  $\bot, a_i, b_j, \top$ will have value $-3$. If there is no such
  pair, any valid sequence contains at most one element other than
  $\bot$ and $\top$ and its value is therefore at least $-2$.
\end{proof}

The proof is in the appendix.

In the rest of the section we give some interesting instantiations of
the subquadratic equivalence of $\Selection$ and $\CLWS$.

\begin{corollary}[$\NestedBoxes \eqquad \VectorDomination$]
\label{cor:nestedboxes}
  The weighted \NestedBoxes problem on $d = c \log n$ dimensions
  can be solved in time $n^{2-(1/\Oh(c \log^2 c))}$. For
  $d = \omega(log n)$, the (unweighted) \NestedBoxes problem cannot
  be solved in time $\Oh(n^{2-\varepsilon})$ for any $\varepsilon > 0$
  assuming $\SETH$.
\end{corollary}
\begin{proof}
  Let $\rel$ be the relation that contains all pairs of non-negative,
  $d$-dimensional vectors $a,b$ such that $a_i \leq b_i$ for all
  $i$. Now $\Selection(\rel)$ is \VectorDomination, and
  $\CLWS(\rel)$ is the \NestedBoxes problem.

  Using the reduction from Theorem~\ref{thm:LWStoSel} and the
  algorithms for vector domination of the stated runtime
  \cite{ImpagliazzoLPS14, Chan15} we immediately get an algorithm for
  \NestedBoxes.

  We apply Theorem \ref{thm:SelToLWS} with $\top = W^d$ where $W$ is the
  largest coordinate in all input vectors, $\bot = 0^d$ and use the
  sum of the coordinates of the boxes as the natural
  ordering. $\SETH$-hardness of \NestedBoxes then follows from
  the $\SETH$-hardness of vector domination \cite{Williams05}.
\end{proof}

If we restrict \NestedBoxes and \VectorDomination to Boolean vectors,
then we get \SubsetChain and \SetContainment respectively. In this
case the upper bound improves to $n^{2-1/\Oh(\log c)}$ \cite{AbboudWY15}.

We would like to point out that the definition of \CLWS requires the
input to be a sequence of data items, and not a set.
Consider the following definition:

\begin{problem}[\ChainSet]
  Let a set of data items data items $\{x_0,\ldots,x_n\}$, weights
  $w_1, \ldots, w_{n-1} \in \{-W,\ldots,W\}$ and a relation
  $\rel(x_i,x_j)$ be given. The chain set problem for $\rel$, denoted
  $\ChainSet(\rel)$ asks to find the weight sequence
  $i_0, i_1, i_2,\ldots, i_k$ such that for all $j$ with
  $1 \leq j \leq k$ the pair $(x_{i_{j-1}},x_{i_j})$ is in the
  relation $\rel$ and the weight $\sum_{j=1}^{k-1} w_{i_j}$ is
  minimized.
\end{problem}

While \CLWS can always be solved in quadratic time, \ChainSet is
$\NP$-complete. For example, consider the box overlap problem: The
input is a set of boxes in two dimensions, given by the top left
corner and the bottom right corner, and the relation consists of all
pairs such that the two boxes overlap. This problem is a
generalization of the Hamiltonian path problem on induced subgraphs of
the two-dimensional grid, which is an $\NP$-complete problem
\cite{ItaiPS82}. This is a formal barrier to a more general reduction
than Theorem \ref{thm:SelToLWS}, as we need some mechanism to impose an
ordering on the data items.

\section{Near-linear time algorithms}
\label{sec:nearlinear}

In this section, we classify problems to be solvable in near-linear
time using the lens of our framework. Note that in these instances,
near-linear time solutions have already been known, however, our focus
on the static variants of LWS provides a simple, general approach to
find fast algorithms by identifying a simple ``core'' problem. Since
in this paper, we generally ignore subpolynomial factors in the
running time, we concentrate here on the reduction from some LWS
variant to its corresponding core problem and disregard reductions in
the other direction.

\subsection{Longest Increasing Subsequence}

The longest increasing subsequence problem \LIS has been first investigated by Fredman~\cite{Fredman75}, who gave an $\Oh(n\log n)$-time algorithm and gave a corresponding lower bound based on \Sorting. The following LWS instantiation is equivalent to \LIS. 
\begin{problem}[\LIS]  We define the following \LWS instantiation $\LIS = \LWS(\weightmat_\lis)$.\\
\Data{integers $x_1, \dots, x_n \in \{1, \dots, W\}$}
\Weights{$w_{i,j} = \begin{cases} -1 & \text{if } x_i < x_j \\ \infty & \text{ow.} \end{cases}$}
\end{problem}

It is straightforward to verify that $-T[n]$ yields the value of the longest increasing subsequence of $x_1, \dots, x_n$. Using the static variant of LWS introduced in Section~\ref{sec:lrlws}, we observe that $\LIS$ effectively boils down to \Sorting.

\begin{obs}\label{obs:lis}
\LIS can be solved in time $\tOh(n)$.
\end{obs}
\begin{proof}
By Lemma~\ref{lem:dynToStatic}, we can reduce \LIS to the static variant $\StaticLWS(\weightmat_\lis)$. It is straight-forward to see that the latter can be reformulated as follows: Given $(a_1, T[1]), \dots, (a_N, T[N])$ and $b_1, \dots, b_N$, determine for every $j=1,\dots, N$, the value $T'[j] = -1+\min_{1\le i \le N, a_i < b_j} T[i]$. To do so, it suffices to sort the first list as $(a_{i_1}, T[i_1]), \dots, (a_{i_N}, T[i_N])$ with $a_{i_1} \le \dots \le a_{i_N}$ and the second as $b_{j_1}, \dots, b_{j_N}$ with $b_{j_1} \le \dots \le b_{j_N}$. Finally, a single pass over both lists will do: For each $k=1,\dots, N$, we search for the largest $\ell$ such that $a_{i_\ell} < b_{j_k}$, then the $T'$-value corresponding to $b_{j_\ell}$ is $-1+ \min_{1\le \ell' \le \ell} T[i_{\ell'}]$. By this approach, it is easy to see that after sorting, these values can be computed in time $\Oh(N)$. For the exact running time, note that solving $\StaticLWS(\weightmat_\lis)$ takes time $\Oh(N\log N)$ due to sorting, yielding a $\Oh(n\log^2 n)$-time algorithm for \LIS by Lemma~\ref{lem:dynToStatic}.
\end{proof}

\subsection{Unbounded Subset Sum}

\USubsetSum is a variant of the classical \SubsetSum, in which repetitions of elements are allowed. While improved pseudo-polynomial-time algorithms for \SubsetSum could only recently be found~\cite{KoiliarisX17, Bringmann17}, there is a simple algorithm solving \USubsetSum in time $\Oh(n \log n)$~\cite{Bringmann17}. It can be cast into an LWS formulation as follows.

\begin{problem}[\USubsetSum]  We define the following \LWS instantiation $\LIS = \LWS(\weightmat_\uss)$.\\
\Data{$S \subseteq [n]$}
\Weights{$w_{i,j} = \begin{cases} 0 & \text{if } j-i \in S \\ \infty & \text{ow.} \end{cases}$}
\end{problem}

Note that in this formulation, $T[n] = 0$ iff there is a multiset of numbers from $S$ that sums up to $n$. It is a straightforward observation that the static variant of \USubsetSum can be solved by classical convolution, i.e., $(\cdot, +)$-convolution.

\begin{obs}\label{obs:subsetsum}
\USubsetSum can be solved in time $\tOh(n)$.
\end{obs}
\begin{proof}
Noting that all weights $w_{i,j}$ are either $0$ or $\infty$, it is easy to see that the static variant $\StaticLWS(\weightmat_\uss)$ can be reformulated as follows: Given a subset $X\subseteq I = \{a+1, \dots, a+N\}$, determine, for all $j \in J = \{a+N+1, \dots, a+2N\}$, whether there exists some $i\in X$ such that $j-i \in S$. To do so, we do the following: We represent $X$ as an $N$-bit vector $x = (x_1, \dots, x_N) \in \{0,1\}^{N}$ with $x_i = 1$ iff. $a+i \in X$. Furthermore, we represent the ``relevant part'' of $S$ by defining a $2N$-bit vector $s = (s_1, \dots, s_{2N}) \in \{0,1\}^{2N}$ with $s_i = 1$ iff. $i \in S$. Then the $(\cdot, +)$-convolution $r = x \classicconv s$ of $x$ and $s$ allows us to determine $T'[a+N+j]$ for $j = 1, \dots, N$: this values is 0 iff $r_{N+j} > 0$ and $\infty$ otherwise. Correctness follows from the observation that $r_{N+j} > 0$ is equivalent to the existence of some $i \in [N]$ and $k\in [2N]$ with $i+k = N+j$ and $x_i = s_k = 1$. This in turn is equivalent to $a+i \in X$ and $(a+N+j)-(a+i) = N+j - i = k \in S$, as desired.

Thus $\StaticLWS(\weightmat_\uss)$ can be solved by a single convolution computation, which can be performed in time $\Oh(N\log N)$. Thus by Lemma~\ref{lem:dynToStatic}, this gives rise to a $\Oh(n\log^2 n)$-time algorithm for $\USubsetSum$.
\end{proof}

\subsection{Concave LWS}

The concave LWS problem is a special case of LWS in which the weights satisfy the quadrangle inequality. Since a complete description of the input instance consists of $\Omega(n^2)$ weights, we use the standard assumption that each $w_{i,j}$ can be queried in constant time. This allows for sublinear solutions in the input description, in particular there exist $\Oh(n)$-time algorithms~\cite{Wilber88,GalilP90}.

\begin{problem}[\ConcaveLWS] We define the following \LWS instantiation $\LIS = \LWS(\weightmat_\concave)$.\\
\Weights{$w_{i,j}$ given by oracle access, satisfying $w_{i,j} + w_{i',j'} \le w_{i',j} + w_{i, j'}$ for $i\le i'\le j \le j'$.}
\end{problem}

We revisit \ConcaveLWS and its known connection to the problem of computing column (or row) minima in a totally monotone\footnote{A matrix $M=(m_{i,j})_{i,j}$ is totally monotone if for all $i<i'$ and $j < j'$, we have that $m_{i,j} > m_{i',j}$ implies that $m_{i,j'} > m_{i',j'}$. For a more comprehensive treatment, we refer to~\cite{SMAWK, GalilP90}.} $(n\times n)$-matrix, which we call the \SMAWKprob because of its remarkable $\Oh(n)$-time solution called the \SMAWK algorithm~\cite{SMAWK}. 

\begin{obs}\label{obs:concavelws}
\ConcaveLWS can be solved in time $\tOh(n)$.
\end{obs}
\begin{proof}
The static variant of \ConcaveLWS can be formulated as follows: Given intervals $I= \{a+1,\dots, a+N\}$ and $J=\{a+N+1, \dots, a+2N\}$, we define a matrix $M := (m_{i,j})_{i\in I ,j\in J})$ with $m_{i,j} = T[i] + w_{i,j}$. It is easy to see that $M$ is a totally monotone matrix since $w$ satisfies the quadrangle inequality. Note that the minimum of column $j\in J$ in $M$ is $\min_{i \in I} T[i]+w_{i,j} = T'[j]$ by definition. Thus, using the SMAWK algorithm we can determine all $T'[j]$ in simultaneously in time $\Oh(N)$.

Thus by Lemma~\ref{lem:dynToStatic}, we obtain an $\Oh(n \log n)$-time algorithm for \ConcaveLWS.
\end{proof}

\myparagraph{Acknowledgments.}  We would like to thank Karl Bringmann
and Russell Impagliazzo for helpful discussions and comments.

\bibliography{lws}

\end{document}